\documentclass[letterpaper,10pt, journal]{IEEEtran}

\IEEEoverridecommandlockouts                               

\usepackage{graphics} 
\usepackage{graphicx}
\usepackage{epsfig} 
\usepackage{amsmath} 
\usepackage{amssymb}  
\usepackage{tabularx}

\usepackage{tikz} 
\usepackage{pgfplots} 

\usepackage{array}
\usepackage{algpseudocode}
\usepackage{algorithm}
\usepackage{epstopdf}
\usepackage{braket}
\usepackage{quantikz}
\usepackage{tikz}
\usepackage{url}
\usetikzlibrary{positioning, arrows.meta}
\usetikzlibrary{shapes.geometric, arrows}
\usetikzlibrary{arrows.meta}
\tikzset{%
  >={Latex[width=2mm,length=2mm]},
            base/.style = {rectangle, rounded corners, draw=black,
                           minimum width=2cm, minimum height=0.5cm,
                           text centered},
  activityStarts/.style = {base, fill=blue!30},
         process/.style = {base, fill=orange!15},
}

\usepackage{color}

\usepackage{lipsum}
\usepackage{verbatim}

\usepackage{amsthm}

\usepackage{mathtools}

\newtheorem{theorem}{\bf Theorem} \newtheorem{definition}[theorem]{\bf Definition} 
 
  \newtheorem{proposition}[theorem]{\bf Proposition} 
  
\newtheorem{Algorithm}[theorem]{\bf Algorithm}

\title{\LARGE \bf
Using quantum computers in control: interval matrix properties
}

\author{Jan Schneider, Julian Berberich
\thanks{The authors are with the Institute for Systems Theory and Automatic Control, University of Stuttgart, 70569 Stuttgart, Germany. 
E-mail: julian.berberich@ist.uni-stuttgart.de, contact for correspondence:
J. Berberich.
Funded by Deutsche Forschungsgemeinschaft (DFG, German Research Foundation) under Germany’s Excellence Strategy - EXC 2075 - 390740016. We acknowledge the support by the Stuttgart
Center for Simulation Science (SimTech).
}
}

\begin{document}
\IEEEpubid{\begin{minipage}{\textwidth}\ \\[12pt] 
    \copyright 2024 EUCA. This version has been accepted for publication in Proc. European Control Conference (ECC), 2024. Personal use of this material is permitted. Permission
from EUCA must be obtained for all other uses, in any current or future media, including reprinting/republishing this material for advertising or promotional
purposes, creating new collective works, for resale or redistribution to servers or lists, or reuse of any copyrighted component of this work in other works.\end{minipage}}

\maketitle


\begin{abstract}
Quantum computing provides a powerful framework for tackling computational problems that are classically intractable.
     The goal of this paper is to explore the use of quantum computers
     for solving relevant problems in systems and control theory. 
In the recent literature, different quantum algorithms have been developed to tackle binary optimization, which plays an important role in various control-theoretic problems. As a prototypical example, we consider the verification of interval matrix properties such as non-singularity and stability on a quantum computer.
     We present a quantum algorithm solving these problems and we study its performance in simulation.
     Our results demonstrate that quantum computers provide a promising tool for control whose applicability to further computationally complex problems remains to be explored.

\end{abstract}

\section{Introduction}
Quantum computing has gained increasing attention in recent years due to its ability to solve certain computationally challenging problems more efficiently than classically possible.
This includes, for example, integer factorization~\cite{shor1999polynomial}, unstructured search~\cite{grover1996fast}, or linear systems of equations~\cite{harrow2009quantum}, but also the simulation of classical~\cite{giannakis2022embedding,schalkers2022efficient} and quantum~\cite{feynman1982simulating}, \cite[Section 4.7]{nielsen2011quantum} systems.
However, these algorithms cannot be implemented reliably for relevant problem sizes on current noisy intermediate-scale quantum (NISQ) devices~\cite{preskill2018quantum} due to problems connected to noise and scalability.
Variational quantum algorithms (VQAs) are a promising class of quantum algorithms which combine a quantum computer with a classical optimization algorithm~\cite{cerezo2021variational}.
VQAs involve trainable parameters which are optimized iteratively using, e.g., a gradient descent scheme, and, therefore, they are well-suited for NISQ devices due to their adaptation to small and noisy hardware.
The quantum approximate optimization algorithm (QAOA)~\cite{farhi2014quantum} is one of the most popular VQAs and it can be used to solve a specific class of integer programs: quadratic unconstrained binary optimization (QUBO) problems.

Problems with integer variables are relevant in various domains of systems and control theory and, therefore, QAOA is a promising candidate for achieving computational speedups in control.
Further quantum algorithms that can tackle computational control problems are listed in~\cite{berberich2023quantum}.
Yet, the usage of quantum computers in control is largely unexplored, except for results on, e.g., MPC with finite input spaces~\cite{inoue2020model} or decentralized control~\cite{deshpande2022quantum}.
In this paper, we propose an algorithm for the verification of interval matrix properties, which is known to be a computationally hard problem, on a quantum computer.

In computational complexity theory, problems can be classified according to the time required for their solution.
Problems which can be solved with a polynomial-time algorithm belong to class P and those whose solution can be verified in polynomial time belong to class NP. For NP-hard problems, which are problems to which every problem in NP can be reduced efficiently, there exist no known polynomial-time algorithms on a classical computer. 
Various control-theoretic problems have been shown to be NP-hard, including static output feedback design or verifying interval matrix properties~\cite{poljak1993checking,nemirovskii1993several,blondel1997np}. 

In this paper, we introduce an approach for verifying interval matrix properties on a quantum computer. We focus on two main properties: robust non-singularity and robust stability of interval matrices, meaning that all members of the interval matrix are non-singular or stable, respectively. Existing research by~\cite{poljak1993checking} shows that the verification of robust non-singularity can be equivalently reformulated as a binary optimization problem. Given that this binary optimization problem is a QUBO problem, it can be tackled using QAOA. The main contribution of this paper is a quantum algorithm based on QAOA which can verify robust non-singularity and robust stability of interval matrices. 
The applicability is illustrated with numerical examples.
\subsection*{\textit{Notation}}
Let $e=\begin{bmatrix}
    1 & \hdots & 1
\end{bmatrix}^\top $ and define the $n$-dimensional discrete cube $Q_n = \{-1,1\}^n$. For each element $z \in Q^n$, the matrix $T_z$ represents a diagonal matrix with the entries of $z$ on its diagonal. Further, for a matrix $A\in\mathbb{R}^{n\times n}$, we define
\begin{align}
    \rho_0(A)=\max \{|\lambda|\;|Ax=\lambda x \; \text{for some} \; x\neq 0,\lambda \in \mathbb{R}\}.
\end{align} 
The Hermitian conjugate of a matrix $A\in\mathbb{C}^{n\times m}$ is denoted by $A^\dagger$.
Finally, $I_n$ denotes an $n\times n$ identity matrix.

\section{Interval Matrices}\label{sec:interval}
In this section, we introduce the concept of interval matrices along with the key properties that are studied in the present paper.
An interval matrix is defined as
\begin{align}
    A^I=[ \,\underline{A},\overline{A}] \,=\{A\in\mathbb{R}^{n\times m}|
\underline{A}\leq A\leq\overline{A}\}
\subseteq\mathbb{R}^{n\times m},
\end{align}
where $\overline{A}$ and $\underline{A}$ represent an upper and lower bound matrix and the inequality is interpreted element-wise. 
Further, we define the center of an interval matrix as 
\begin{align}
A_m = \frac{1}{2}(\overline{A}+\underline{A}).
\end{align}
\vskip20pt
\newpage

This allows to reformulate the lower and upper bound matrices as $\overline{A}=A_m+\Delta_d$ and $\underline{A}=A_m-\Delta_d$ with $\Delta_d=\frac{1}{2}(\overline{A}-\underline{A})$.

\subsection{Non-singularity of interval matrices}
In this paper, we study
non-singularity of interval matrices based on 
the radius of non-singularity. The definition of a non-singular interval matrix is given in the following.
\begin{definition}
An interval matrix $A^I$ is non-singular if all matrices $A\in A^I$ are non-singular.
\end{definition}
The radius of non-singularity measures the distance of a matrix $A$ to the closest singular matrix according to an a priori fixed shifting matrix $\Delta$ and is defined as
\begin{align}
\label{eq:d}
    d(A,\Delta)=\min_{\varepsilon\geq0}\quad&\varepsilon\\\nonumber
\text{s.t.}\quad& A-\varepsilon\Delta\leq A'\leq A+\varepsilon\Delta\\ \nonumber
&\text{for some singular}\>A'\in\mathbb{R}^{n\times n}.
\end{align}
This problem minimizes $\varepsilon$ such that the interval matrix $[A-\varepsilon\Delta,A+\varepsilon\Delta]$ is singular, i.e., contains a singular matrix $A'\in[A-\varepsilon\Delta,A+\varepsilon\Delta]$. 
Given the optimal value $\varepsilon^*$ of problem~\eqref{eq:d}, a non-singular interval matrix can be defined accordingly as
\begin{align}
    [A-\varepsilon\Delta,A+\varepsilon\Delta]
\end{align}
for any $\varepsilon < \varepsilon^*$. 
The radius of non-singularity is relevant, e.g., for sensor processing in the presence of noise~\cite{poljak1993checking} or for sensitivity analysis of linear systems~\cite{deif1986sensitivity}, and it has close connections to various concepts including the structured singular value which is a powerful tool in robust stability analysis~\cite{doyle1982analysis}.
Further, as we show below, it can be used to study robust stability properties of uncertain linear systems.

As an instructive example, let us consider the matrices 
\begin{align}
\label{eq:example}
    A = \begin{bmatrix}
        1 & -1 \\ 0 & 1
    \end{bmatrix},
    \Delta = \begin{bmatrix}
        1 & 1 \\ 1 & 1
    \end{bmatrix}.
\end{align}
The boundaries of the corresponding interval matrix considered in~\eqref{eq:d} are given by 
\begin{align}
    A\pm\varepsilon\Delta=\begin{bmatrix}
        1\pm\varepsilon & -1 \pm\varepsilon \\ \pm\varepsilon & 1\pm\varepsilon
    \end{bmatrix}.
\end{align}
The optimal solution of~\eqref{eq:d} is given by $\varepsilon^*=\frac{1}{3}$ with the associated singular matrix
\begin{align}
    A'=A-\frac{1}{3}\Delta=\begin{bmatrix}
        \frac{2}{3} & -\frac{4}{3} \\ -\frac{1}{3} & \frac{2}{3}
    \end{bmatrix}.
\end{align}
As shown in~\cite{poljak1993checking}, the radius of non-singularity of a given, non-singular matrix $A$ and a matrix $\Delta$, can be reformulated as a combinatorial optimization problem:
\begin{align}
    \label{eq:d_re}
    d(A,\Delta)=\frac{1}{\max\{\rho_0(A^{-1}T_y \Delta T_z)|y,z \in Q_n \}}.
\end{align}
The optimization problem appearing in the denominator is the central problem considered in this paper. The proof of NP-hardness is given in \cite{poljak1993checking} and is sketched briefly in the following.
The maximization in~\eqref{eq:d_re} can be simplified for the specific choice $\Delta=ee^\top $, which allows to reformulate~\eqref{eq:d_re} as
\begin{align}
\label{eq:simpled}
    d(A) = \frac{1}{r(A^{-1})}
\end{align}
with
\begin{align}
    \label{eq:r_min}
    r(M)&=\max\{z^\top My|z,y\in Q_n\}.
\end{align}

It is shown in \cite{poljak1993checking} that the \emph{maximum-cut problem} can be reduced to the optimization problem~\eqref{eq:r_min}. Since the maximum-cut problem is NP-hard, it follows that computing the radius of non-singularity is NP-hard as well.

In this work, we derive a combinatorial optimization problem similar to~\eqref{eq:simpled} but for a more general class of matrices $\Delta$ than considered in \cite{poljak1993checking}. To be precise, the condition $\Delta=ee^\top $ is relaxed to $\Delta$ having rank $1$, allowing for more versatile and practically relevant problem considerations. 
\subsection{Stability of interval matrices}
Beyond non-singularity, we also address the problem of stability verification of interval matrices.
\begin{definition}
An interval matrix $A^I$ is stable if all matrices $A\in A^I$ are stable, i.e., $\mathrm{Re}(\lambda_i(A))<0$ for all eigenvalues $\lambda_i(A)$ of $A$.
\end{definition}

Stability of an interval matrix $A^I$ implies robust stability of the uncertain system
\begin{align}
    \dot{x}=Ax
\end{align}
with $A\in A^I$, which is an important problem in robust control~\cite{zhou1996robust}.
In \cite{rohn1994checking}, it is shown that, just like non-singularity, verifying stability of interval matrices is NP-hard as well. 
In the following result, we exploit that stability and non-singularity are closely related due to the continuity of eigenvalues in the matrix entries. 
For simplicity, we focus on symmetric interval matrices.
Considering general interval matrices is an interesting issue for future research.
\begin{definition}
    For an interval matrix $A^I$, the symmetric interval matrix $A^I_{\mathrm{sym}}$ is defined as 
    \begin{align}
        A^I_{\mathrm{sym}}=\{A\in A^I\mid A=A^\top\}.
    \end{align}
\end{definition}
Stability verification of a symmetric interval matrix is relevant in scenarios, where not only uncertainty bounds specified by $A^I$ are available but also additional structural knowledge on symmetry of the dynamics.
The following result shows that stability of symmetric interval matrices can be verified based on a non-singularity test.
\begin{proposition}\label{prop:1}
Suppose $A^I_{\mathrm{sym}}$ is a non-singular symmetric interval matrix and there exists a stable $A\in A^I_{\mathrm{sym}}$.
Then, $A^I_{\mathrm{sym}}$ is stable.
\end{proposition}
\textbf{Proof.}
If $A\in A^I_{\mathrm{sym}}$ is stable, then all its eigenvalues have real value less than $0$. Note that all eigenvalues of the matrices in $A^I_{\mathrm{sym}}$ are real due to symmetry. Since eigenvalues are continuous functions in the matrix entries, the non-singularity of $A^I_{\mathrm{sym}}$ implies that all $A'\in A^I_{\mathrm{sym}}$ are stable.
\qed

Motivated by Proposition \ref{prop:1}, we focus on deriving a quantum algorithm for verifying robust non-singularity, which then also allows to verify robust stability of a symmetric interval matrix by testing stability of one of its elements. Note that the eigenvalues of a single matrix can be computed efficiently on a classical computer, which is the only additional step for the stability check.
\section{Quantum Algorithm}
In this section, we propose the quantum algorithm for verifying the interval matrix properties introduced above. First, in Section~\ref{subsec:qb}, we present basic elements of quantum computing that are required to state and implement our algorithm. Next, in Section~\ref{subsec:QAOA}, the key underlying algorithm QAOA is explained in more detail followed by the construction of the problem Hamiltonian in Section~\ref{subsec:problem_Hamiltonian}. Finally, the overall algorithm is stated in Section~\ref{subsec:algo}.
\subsection{Quantum Basics}
\label{subsec:qb}
In the following, we present required basics of quantum computing.
For further details, we refer to the tutorial~\cite{berberich2023quantum}, which introduces quantum computing from the perspective of control, as well as to the textbook~\cite{nielsen2011quantum}.
Qubits (short for ``quantum bits'') are the basic components of a quantum computer, comparable to bits in classical computing. Mathematically, a qubit is written as
\begin{equation}\label{eq:qubit}
    \ket{\psi}=\begin{bmatrix} z_0 \\ z_1
    \end{bmatrix}
\end{equation}
for some $z_0,z_1 \in \mathbb{C}$ with $|z_0|^2+|z_1|^2=1$.
We use the standard Dirac notation $\ket{\psi}$ for denoting quantum states and $\bra{\psi}=\ket{\psi}^{\dagger}$ for their Hermitian conjugate.
The precise value of $\ket{\psi}$ is not accessible.
Instead, a measurement can only reveal that the system is in one of the two basis states 
\begin{align}
    \ket0=\begin{bmatrix}1\\0\end{bmatrix}\>\text{and}\>\ket1=\begin{bmatrix}0\\1\end{bmatrix},
\end{align}
where the probabilities for either outcome are  $|z_0|^2$ and $|z_1|^2$, respectively.
A quantum state consisting of $n$ qubits lives in $\mathbb{C}^{2^n}.$

The second fundamental component of a quantum computer are quantum gates. 
These are used to manipulate the qubits and they are represented by unitary matrices $U$, i.e., they satisfy $U^\dagger U=I$.
The application of a quantum gate $U$ to an input state $\ket{\psi_{\mathrm{in}}}$ is defined by multiplication, i.e., the output state is 
\begin{align}
\ket{\psi_{\mathrm{out}}}=U\ket{\psi_{\mathrm{in}}}.
\end{align}
The \emph{Pauli gates} $X$, $Y$, and $Z$ are popular examples of quantum gates and they are defined as
\begin{align}
    X = \begin{bmatrix} 0 & 1 \\ 1 & 0 
    \end{bmatrix},
    Y = \begin{bmatrix} 0 & -i \\ i & 0 
    \end{bmatrix},
    Z = \begin{bmatrix} 1 & 0 \\ 0 & -1 
    \end{bmatrix}.
\end{align}
For example, applying the $X$ gate to a qubit swaps the amplitudes of $\ket{0}$ and $\ket{1}$ as can be seen by 
\begin{align}
\label{eq:X01}
    X \ket{0} = \begin{bmatrix} 0 & 1 \\ 1 & 0 
    \end{bmatrix} 
    \begin{bmatrix}  1 \\  0 
    \end{bmatrix} = \ket{1}
\end{align}
and similarly $X\ket1=\ket0$.
The final part of a quantum algorithm is the measurement.
In Figure~\ref{fig:1qex}, a simple example of a quantum circuit is given which implements the operation \eqref{eq:X01} with subsequent measurement of the qubit. 

\begin{figure}
    \centering
    \begin{quantikz}
\lstick{$\ket{0}$} & \gate{X} &
    \meter{}
\rstick{}
    \end{quantikz}
    \caption{Quantum circuit implementing the operation~\eqref{eq:X01} followed by a measurement.}
    \label{fig:1qex}
\end{figure}
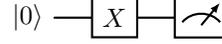

As mentioned above, qubits cannot be measured directly, but a measurement can only return one of a finite number of possible outcomes, where the corresponding probabilities are determined by the probability amplitudes $z_0$, $z_1$ of the qubit, cf.\ \eqref{eq:qubit}.
To be more precise, measurements are always taken w.r.t.\ an \emph{observable} $\mathcal{M}=\mathcal{M}^\dagger$.
Repeated measurements of a state $\ket{\psi}$ w.r.t.\ the observable $\mathcal{M}$ allow to retrieve the expectation value 
\begin{align}\label{eq:measurement}
    \braket{\psi|\mathcal{M}|\psi}=\psi^\dagger\mathcal{M}\psi.
\end{align}
For many quantum algorithms, in particular VQAs, this expectation value is the actual output of the algorithm. 

\subsection{Quantum approximate optimization algorithm}
\label{subsec:QAOA}
VQAs are a class of quantum algorithms 
combining an optimization procedure on a classical computer with a quantum algorithm~\cite{cerezo2021variational}. This allows the algorithm to adapt to noise on the quantum computer and, thereby, provide possibly more reliable results.
QAOA~\cite{farhi2014quantum} is an important example of VQAs. It is tailored to solving combinatorial optimization problems and involves alternating between a classical optimization routine to determine parameters of quantum gates and a quantum evolution via the resulting parametrized circuit.

In the following, we provide a brief introduction to QAOA.
QAOA can be used to solve problems of the form
\begin{align}
    \label{eq:opti1}
    \max_{x\in\{0,1\}^n} \quad&C(x)
\end{align}
for a cost function $C:\{0,1\}^n\to\mathbb{R}$.
The algorithm consists of repeated applications of parametrized unitaries to an initial quantum state $\ket{\psi_0}=\begin{bmatrix}
    1&\hdots&1
\end{bmatrix}^\top$.
The two unitaries applied to this state are the \emph{mixing unitary} $U_{\mathrm{B}}(\beta_j)=e^{-i\beta_j H_{\mathrm{B}}}$ with mixing Hamiltonian $H_{\mathrm{B}}=H_{\mathrm{B}}^\dagger$ and the \emph{problem unitary} $U_{\mathrm{P}}(\gamma_j)=e^{-i\gamma_j H_{\mathrm{P}}}$ with problem Hamiltonian $H_{\mathrm{P}}=H_{\mathrm{P}}^\dagger$. The parameter vectors $\beta,\gamma \in \mathbb{R}^p$ are optimization variables for the $p$ layers of the quantum circuit.
These two unitaries are applied to the initial state $\ket{\psi_0}$ in an alternating fashion
\begin{align}\label{eq:QAOA_circuit}
    U(\beta_p,\gamma_p)\cdots U(\beta_1,\gamma_1) 
\end{align}
with $U(\beta_j,\gamma_j)=U_{\mathrm{B}}(\beta_j)U_{\mathrm{P}}(\gamma_j)$.
The number of alternations $p$ in~\eqref{eq:QAOA_circuit} is problem-specific but typically chosen small to enable implementations on current quantum hardware.
Figure~\ref{fig:circuit_QAOA} shows the $2$-qubit quantum circuit consisting of a measurement of the state
\begin{align}
    U(\beta_2,\gamma_2)U(\beta_1,\gamma_1)\ket{\psi_0}.
\end{align}
Here, $H$ denotes the \emph{Hadamard} gate defined as 
\begin{align}
H=\frac{1}{\sqrt{2}}\begin{bmatrix}
    1&1\\1&-1
\end{bmatrix},
\end{align}
which generates the input state due to 
\begin{align}
    \ket{\psi_0}=(H\otimes H)(\ket0\otimes\ket0)
\end{align}
with the Kronecker product $\otimes$.

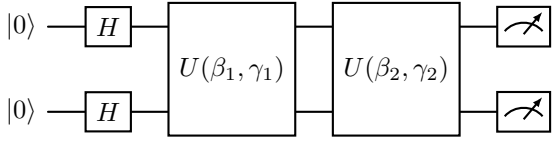
\begin{figure}
    \centering
    \begin{quantikz} \lstick{$\ket{0}$}&\gate{H}&\gate[2]{U(\beta_1,\gamma_1)}&\gate[2]{U(\beta_2,\gamma_2)}& \meter{} \\ \lstick{$\ket{0}$}&\gate{H}&&&\meter{}
    \end{quantikz}
    \caption{Circuit representation of QAOA for $n=2$ qubits and $p=2$ layers.}
    \label{fig:circuit_QAOA}
\end{figure}

Let us now discuss how the mixing and problem Hamiltonians $H_{\mathrm{B}}$ and $H_{\mathrm{P}}$ are chosen.
The problem Hamiltonian encodes the 
cost function $C(x)$ into the quantum algorithm, see Section~\ref{subsec:problem_Hamiltonian} for the construction procedure.
On the other hand, the mixing Hamiltonian is used for entangling the solution amplitudes and expanding the solution space to create a diverse set of candidate solutions. 
The mixing Hamiltonian is commonly chosen independently of the cost function $C(x)$ as
\begin{align}
    H_{\mathrm{B}} = \sum_{j=1}^n X_j,
\end{align}
where we use the standard notation
\begin{align}\label{eq:Xi_def}
    X_j=\underbrace{I_2\otimes\cdots\otimes I_2}_{j-1\>\text{times}}
    \otimes\, X\otimes \underbrace{I_2\otimes\cdots\otimes I_2}_{n-j\>\text{times}}.
\end{align}
Finally, the expectation value 
\begin{align}\label{eq:QAOA_expectation}
\braket{\psi(\beta,\gamma)|H_{\mathrm{P}}|\psi(\beta,\gamma)}
=\psi(\beta,\gamma)^{\dagger} H_{\mathrm{P}}\psi(\beta,\gamma)
\end{align}
of the parametrized state $\ket{\psi(\beta,\gamma)}$ is determined based on repeated measurements, compare~\eqref{eq:measurement}. 
Based on this value, a classical computer determines a new set of parameters $\gamma$ and $\beta$ maximizing~\eqref{eq:QAOA_expectation} using some optimization scheme, e.g., gradient descent.
This procedure is then repeated until convergence.
Figure~\ref{fig:VQA} summarizes the basic scheme of 
QAOA as an iteration between classical optimization and execution of the quantum algorithm~\eqref{eq:QAOA_expectation}.

After converging to a set of parameters $\beta^*$, $\gamma^*$, the candidate solution for~\eqref{eq:opti1} is obtained by performing repeated measurements of the resulting quantum state
\begin{align}\label{eq:QAOA_output_state}
    &\ket{\psi(\beta^*,\gamma^*)}\\\nonumber 
    =&a_1\ket{0\dots00}+a_2\ket{0\dots01}+\dots+a_{2^n}\ket{1\dots11},
\end{align}
where we use the standard notation for \emph{computational basis states}, e.g.,
\begin{align}
\ket{0...011}=\ket0\otimes\dots\otimes\ket0\otimes\ket1\otimes\ket1.
\end{align}
A measurement of $\ket{\psi(\beta^*,\gamma^*)}$ returns the $j$-th computational basis state with probability $|a_j|^2$.
The computational basis state with the highest probability (determined empirically by counting occurrences in repeated measurements) is then employed as candidate solution for the original binary optimization problem~\eqref{eq:opti1}.

\begin{figure}
    \centering
    \begin{tikzpicture}[node distance=1.5cm,
    every node/.style={fill=white, font=\small}, align=center]
  \node (paramUpdate)             [activityStarts]              {Parameter Update\\on Classical Computer};
  \node (quantEval)     [process, below of=paramUpdate]          {Evaluating Cost Function\\on Quantum Computer};
  \draw[->] (paramUpdate.east) -- ++(1.5,0) -- ++(0,-1.5) --                
     node[xshift=0.6cm,yshift=0.8cm, text width=2.0cm]
     {$\beta,\gamma$}(quantEval.east);  
  \draw[<-] (paramUpdate.west) -- ++(-1.5,0) -- ++(0,-1.5) --               
     node[xshift=-1cm,yshift=0.8cm, text width=2.0cm]
     {$\psi(\beta,\gamma)^{\dagger} H_{\mathrm{P}}\psi(\beta,\gamma)$}(quantEval.west);  
   
  \end{tikzpicture}
    \caption{Basic scheme of QAOA as an iterative optimization containing classical parameter updates and executions of parametrized quantum circuits.}
    \label{fig:VQA}
\end{figure}
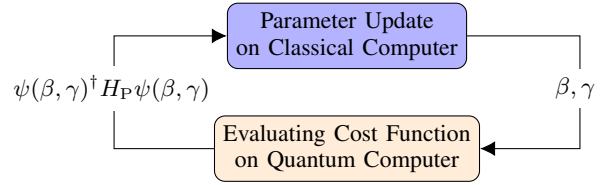

\subsection{Formulating the problem Hamiltonian}\label{subsec:problem_Hamiltonian}

In the following, we show how the problem Hamiltonian $H_{\mathrm{P}}$ can be chosen in order to compute the radius of non-singularity via QAOA.
Recall the combinatorial optimization problem introduced in \eqref{eq:d_re}, i.e.,
\begin{align}
\label{eq:opti}
\max\; \bar{C}(y,z)=\max\{\rho_0(A^{-1}T_y \Delta T_z)|y,z \in Q_n \}.
\end{align}
In order to apply QAOA, we now transform~\eqref{eq:opti} into a binary optimization problem via a variable transformation from $\pm1$ to $0/1$.
More precisely, we introduce the binary variable $x\in\{0,1\}^n$ according to
$x_i=\frac{1-z_i}{2}$ and $x_{i+n}=\frac{1-y_i}{2}$, resulting in
\begin{align}
\label{eq:CE}
    (-1)^{x_i}=z_i,\> (-1)^{x_{i+n}}=y_i.
\end{align}
In order to construct a suitable $H_{\mathrm{P}}$, the relation
\begin{align}
    \label{eq:hp}
    C(x)\ket{x}=H_{\mathrm{P}}\ket{x}
\end{align}
needs to hold for any $x\in\{0,1\}^n$.
Here, $C(x)$ is the cost in the new parametrization such that $C(x)=\bar{C}(y,z)$ for all $x\in\{0,1\}^n$ and $y,z$ satisfying~\eqref{eq:CE}.
Indeed, multiplying~\eqref{eq:hp} by $\bra{x}$ from the left-hand side and using unit norm of $\ket{x}$ yields $C(x)=\braket{x|H_{\mathrm{P}}|x}$, which is why QAOA aims at maximizing $\braket{x|H_{\mathrm{P}}|x}$, compare~\eqref{eq:QAOA_expectation}. 

The following result provides a possible choice of $H_{\mathrm{P}}$ satisfying~\eqref{eq:hp}.
\begin{theorem}\label{thm:H_P}
\label{th:hp}
Suppose $\mathrm{rank}(\Delta)=1$ and define the problem Hamiltonian
    \begin{align*}
        H_{\mathrm{P}}=\sum_{i,j=1}^n v_i\Tilde{a}_{ij} \delta_jZ_iZ_{n+j}
    \end{align*} 
with $(\tilde{a}_{ij})=A^{-1}$, $Z_i$ as in~\eqref{eq:Xi_def} for $X=Z$, and $\delta,v\in\mathbb{R}^{n}$ such that $\Delta=\delta v^\top$.
Then,~\eqref{eq:hp} holds.
\end{theorem}
\begin{proof}
First,~\eqref{eq:opti} is reformulated using that $\Delta=\delta v^\top$, i.e., $\mathrm{rank}(\Delta)=1$. To this end, if $\lambda$ is a non-zero, real eigenvalue of $A^{-1}T_y\Delta T_z$, it holds that
\begin{align}
\label{eq:lw}
    A^{-1}T_y\Delta T_z w = A^{-1}T_y\delta v^\top  T_z w = \lambda w.
\end{align}
Since $\lambda \neq 0$, it follows that $v^\top  T_z w \neq 0$.
Together with \eqref{eq:lw}, this implies
\begin{align}
    A^{-1}T_y\delta &= \lambda w (v^\top  T_z w)^{-1}. 
\end{align}
Left-multiplying by $v^\top  T_z$ yields
\begin{align}
    v^\top  T_z A^{-1}T_y\delta &= \lambda.
\end{align}
Subsequently, this results in
\begin{align}
    \label{eq:finmax}
    \max \bar{C}(y,z)&=\max\{|v^\top  T_z A^{-1}T_y\delta|\;|y,z \in Q_n \} \\
    &=\max\{v^\top  T_z A^{-1}T_y\delta|y,z \in Q_n \}. \nonumber
\end{align}
The absolute value can be neglected due to the (bi-)linear cost function and the symmetry of $Q_n$.
Substituting \eqref{eq:CE} into $\bar{C}(y,z)$, we obtain
\begin{align}
\label{eq:sumcost}
    C(x)&=v^\top  T_{x_z} A^{-1}T_{y_z} \delta \\
    &= \begin{bmatrix}
        x_{z,1}v_{1}\\
        \vdots \\
        x_{z,n}v_{n}
    \end{bmatrix}^\top 
        A^{-1}\begin{bmatrix}
        x_{y,1}\delta_{1}\\
        \vdots \\
        x_{y,n}\delta_{n}
    \end{bmatrix} \nonumber \\\nonumber
    &=\sum_{i,j=1}^n v_i\Tilde{a}_{ij} \delta_j (-1)^{x_i} (-1)^{x_{n+j}}
\end{align}
with 
\begin{align}
    x_z&=\begin{bmatrix}
        (-1)^{x_1} & \hdots & (-1)^{x_n}
    \end{bmatrix}^\top  \\
    x_y&=\begin{bmatrix}
        (-1)^{x_{n+1}} & \hdots & (-1)^{x_{2n}}
    \end{bmatrix}^\top.  \nonumber
\end{align}
Note that the Pauli-Z gate satisfies
\begin{align}
    Z\ket{x}=(-1)^x\ket{x}, x\in \{0,1\}.
\end{align}
Applied to~\eqref{eq:sumcost}, this results in 
\begin{align}
    C(x)\ket{x}&=\sum_{i,j=1}^n v_i\Tilde{a}_{ij} \delta_jZ_iZ_{n+j}\ket{x_0,\hdots,x_{2n-1}} \\
    &=H_{\mathrm{P}}\ket{x}. \nonumber\qedhere
\end{align}
\end{proof}

The proof of Theorem~\ref{thm:H_P} consists of two parts:
1) reformulation of the optimization problem~\eqref{eq:opti} as a QUBO and 2) reformulation of the QUBO as a VQA by choosing a suitable problem Hamiltonian $H_{\mathrm{P}}$.
The step 1) is inspired by~\cite{poljak1993checking} but faces the additional technical challenge of considering a general rank-$1$ matrix $\Delta$ rather than $\Delta=ee^\top$.
Further, the step 2) follows the derivation in~\cite{hadfield2021representation}, adapted to the specific binary optimization problem considered in the present paper.

\subsection{The Quantum Algorithm}
\label{subsec:algo}
Algorithm~\ref{alg:rns} summarizes the overall quantum algorithm verifying robust non-singularity of a given interval matrix $A^I$. It can also be used to verify stability of a symmetric interval matrix by using Proposition~\ref{prop:1}.

\begin{algorithm}
\begin{Algorithm}\label{alg:rns}
\normalfont{\textbf{Verifying non-singularity of an interval matrix}}\\
Input: $A^I=[A_m-\varepsilon\Delta,A_m+\varepsilon\Delta] \,$ and $\Delta=(\delta_{ij})$ with $\delta_{ij}\geq0$ and $\text{rank}(\Delta)=1$.
\begin{enumerate}
\item Compute $A_m^{-1}$
\item Construct problem Hamiltonian $H_{\mathrm{P}}$ (cf.\ Theorem~\eqref{thm:H_P})
\item Run QAOA to compute $d(A_m,\Delta)$ (cf.\ Section \ref{subsec:QAOA} and \cite{farhi2014quantum} for details)
\item If $d(A_m,\Delta)> \varepsilon$, then $A^I$ is non-singular
\end{enumerate}
\end{Algorithm}
\end{algorithm}

\section{Implementation}
In order to examine the performance of the proposed algorithm, we simulate the algorithm using the Pennylane toolbox \cite{bergholm2018pennylane} for Python.
The source code for the implementation is publicly accessible on \url{https://github.com/JanKyb/Radius-Of-NonSingularity-using-QAOA}.

To verify the performance, we consider two examples. 
For the first example, the matrices of \eqref{eq:example} are reconsidered. 
Calculating the radius of non-singularity with the proposed quantum algorithm leads to the results in Figure~\ref{fig:2dimex}. 
The counts shown in the figure represent an approximation of the probabilities $|a_j|^2$ corresponding to the quantum state~\eqref{eq:QAOA_output_state} obtained via the proposed algorithm, followed by an additional conversion step according to~\eqref{eq:CE}.

\begin{figure}
    \centering
    \includegraphics[scale=0.6]{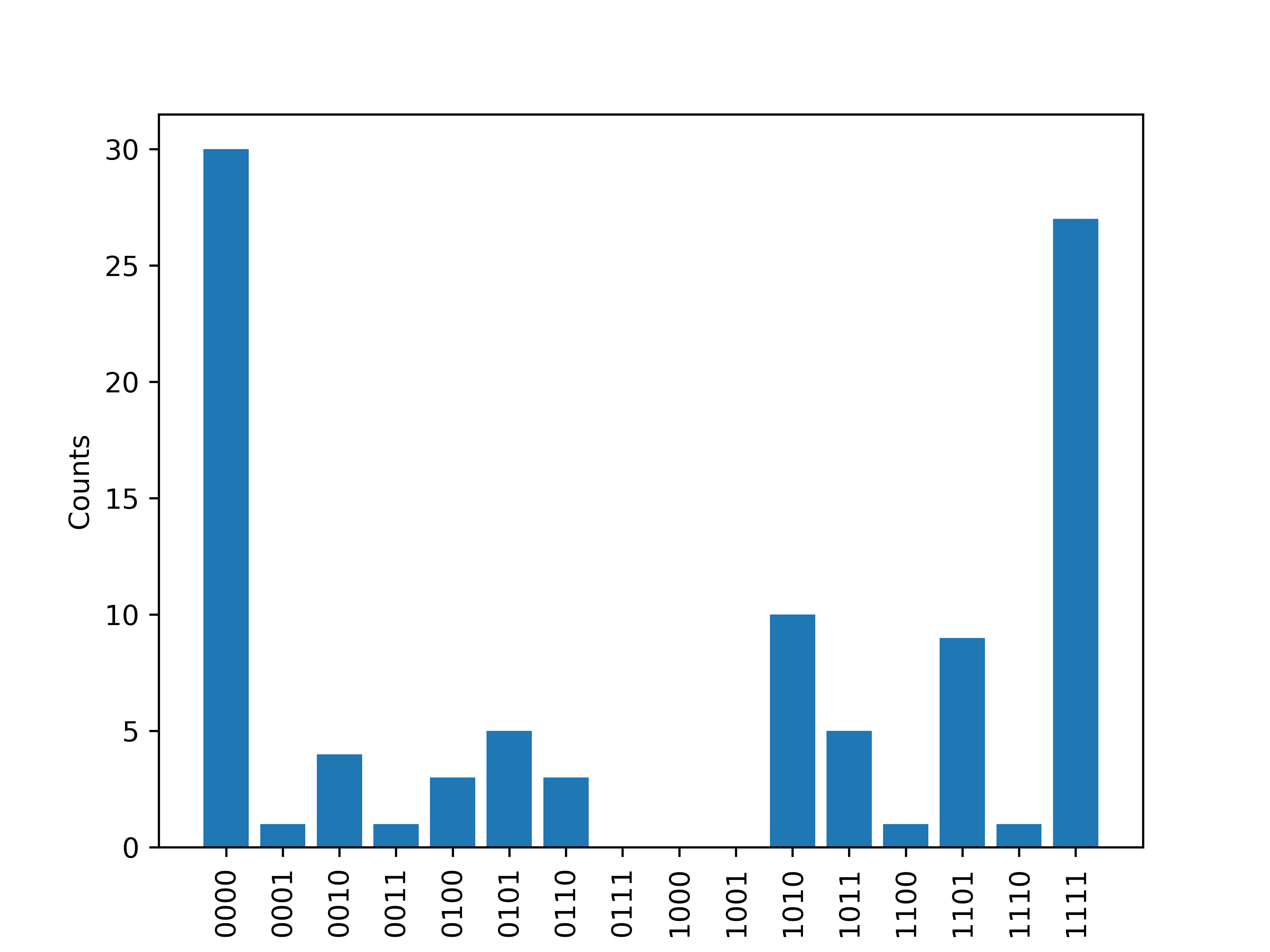}
    \caption{Counts per string for the output of the proposed algorithm when applied to the first example.}
    \label{fig:2dimex}
\end{figure}

It can be seen, that the first string $0000$ yields the solution with the (empirically) highest probability, i.e., the maximum amount of counts. 
Transforming this string back to the initial coordinates based on~\eqref{eq:CE}, the corresponding candidate solution of~\eqref{eq:opti} is given by $z=\begin{bmatrix}
    1 & 1
\end{bmatrix}^\top $ and $y=\begin{bmatrix}
    1 & 1
\end{bmatrix}^\top $.
This yields the result $d(A,\Delta)=\frac{1}{3}$, which is indeed the correct solution, compare Section~\ref{sec:interval}.
Note that also the solution string $1111$ has a comparably large amount of counts. Evaluating the cost function, $1111$ also yields the same optimal result as $0000$, which is due to the symmetry in the variable transformation~\eqref{eq:CE} and the cost function in~\eqref{eq:opti}.
Moreover, several further strings corresponding to a suboptimal solution also have a non-trivial amount of counts. This characteristic is representative for QAOA being a quantum algorithm with an inherently probabilistic output.

As a second and more realistic example, we study robust stability of an RL circuit from~\cite{meisami2008dissipative}, which is a prime example of a symmetric system.
The corresponding system matrix is
\begin{align}
    A=\begin{bmatrix}
        -2&2&0\\2&-5&3\\0&3&-7
    \end{bmatrix}.
\end{align}
To analyze robust stability of this system subject to additional uncertainty, we are interested in finding a possibly large value $\varepsilon$ such that all symmetric matrices in 
$[A-\varepsilon\Delta,A+\varepsilon\Delta]$
are stable, where we consider $\Delta=\begin{bmatrix}
    1&1&1\end{bmatrix}\begin{bmatrix}
        1&1&1\end{bmatrix}^\top$.
In the following, this analysis will be carried out by computing the radius of non-singularity $d(A,\Delta)$ via the proposed quantum algorithm.

        An application of Algorithm~\ref{alg:rns} leads to a distribution of counts analogous to Figure~\ref{fig:2dimex}.
        Due to the increased dimension, we do not depict all possible solutions but only list the bit strings with the most counts in Table~\ref{tab:counts}.

        \begin{table}
        \begin{center}
        \begin{tabular}{ | m{6em} |  m{1cm} |} 
          \hline
          binary strings  & counts \\ 
          \hline
          \hline
          $000111$  & $48$  \\ 
          \hline
          $111000$ & $33$  \\ 
          \hline
          $111001$  & $5$  \\ 
          \hline
          $000110$  & $2$  \\ 
          \hline
          $010011$  & $2$  \\ 
          \hline
        \end{tabular}\vskip10pt
        \caption{\label{tab:counts}Counts per String for the output of the proposed algorithm when applied to the second example}
        \end{center}
        \end{table}

        The solution with the highest number of counts is $000111$.
        To verify that this is indeed an optimal solution, we use~\eqref{eq:CE} to transform the optimal string $000111$ into the initial coordinates as $z=\begin{bmatrix}
            1&1&1
        \end{bmatrix}^\top$ and $y=\begin{bmatrix}
            -1&-1&-1
        \end{bmatrix}^\top$.
        Plugging this candidate into the cost of~\eqref{eq:opti} yields $4.0833$.
        Thus, we have $d(A,\Delta)=\frac{1}{4.0833}$ and, indeed, $A-\frac{1}{4.0833}\Delta$ is (approximately, i.e., modulo numerical inaccuracies) singular.
Using that $A$ is stable, Proposition~\ref{prop:1} implies that all symmetric matrices in $[A-\varepsilon\Delta,A+\varepsilon\Delta]$ for any $\varepsilon<4.0833$ are stable.

Finally, note that the candidate $111000$ only has (compared to the other candidates) slightly less counts than the best solution $000111$ since both candidates are, in fact, equivalent.
This follows again from the symmetry in the variable transformation~\eqref{eq:CE} and the cost function in~\eqref{eq:opti}.

\section{Conclusion}

Quantum computing is a rapidly advancing technology that promises to solve certain computational problems faster than classically possible. 
In this paper, we presented a quantum algorithm for verifying non-singularity and stability of interval matrices, which are relevant problems, e.g., in robust stability analysis.
The proposed algorithm relies on QAOA which is a popular recent quantum algorithm addressing combinatorial optimization.
Extending our results to stability analysis of general (not symmetric) interval matrices as well as the implementation on a real quantum computer are interesting issues for future research.
Moreover, given the high relevance of combinatorial optimization problems in various domains of control, applying QAOA to solve computationally complex problems in control poses another promising future research direction.
Beyond combinatorial optimization, further computational problems appearing in control may as well be amenable to quantum computing, see~\cite{berberich2023quantum} for an overview.

\bibliographystyle{IEEEtran}  
\bibliography{sample}

\begin{thebibliography}{10}
\providecommand{\url}[1]{#1}
\csname url@samestyle\endcsname
\providecommand{\newblock}{\relax}
\providecommand{\bibinfo}[2]{#2}
\providecommand{\BIBentrySTDinterwordspacing}{\spaceskip=0pt\relax}
\providecommand{\BIBentryALTinterwordstretchfactor}{4}
\providecommand{\BIBentryALTinterwordspacing}{\spaceskip=\fontdimen2\font plus
\BIBentryALTinterwordstretchfactor\fontdimen3\font minus \fontdimen4\font\relax}
\providecommand{\BIBforeignlanguage}[2]{{%
\expandafter\ifx\csname l@#1\endcsname\relax
\typeout{** WARNING: IEEEtran.bst: No hyphenation pattern has been}%
\typeout{** loaded for the language `#1'. Using the pattern for}%
\typeout{** the default language instead.}%
\else
\language=\csname l@#1\endcsname
\fi
#2}}
\providecommand{\BIBdecl}{\relax}
\BIBdecl

\bibitem{shor1999polynomial}
P.~Shor, ``Polynomial-time algorithms for prime factorization and discrete logarithms on a quantum computer,'' \emph{SIAM Review}, vol.~41, no.~2, pp. 303--332, 1999.

\bibitem{grover1996fast}
L.~K. Grover, ``A fast quantum mechanical algorithm for database search,'' in \emph{Proc. 28th ACM Symposium on the Theory of Computing}, 1996, pp. 212--219.

\bibitem{harrow2009quantum}
A.~W. Harrow, A.~Hassidim, and S.~Lloyd, ``Quantum algorithm for linear systems of equations,'' \emph{Physical Review Letters}, vol. 103, no.~15, p. 150502, 2009.

\bibitem{giannakis2022embedding}
D.~Giannakis, A.~Ourmazd, P.~Pfeffer, J.~Schumacher, and J.~Slawinska, ``Embedding classical dynamics in a quantum computer,'' \emph{Physical Review A}, vol. 105, p. 052404, 2022.

\bibitem{schalkers2022efficient}
M.~A. Schalkers and M.~M{\"o}llers, ``Efficient and fail-safe collisionless quantum {Boltzmann} method,'' \emph{arXiv:2211.14269}, 2022.

\bibitem{feynman1982simulating}
R.~P. Feynman, ``Simulating physics with computers,'' \emph{Int. J. Theor. Phys.}, vol.~21, p. 467, 1982.

\bibitem{nielsen2011quantum}
M.~A. Nielsen and I.~L. Chuang, \emph{Quantum Computation and Quantum Information: 10th Anniversary Edition}, 10th~ed.\hskip 1em plus 0.5em minus 0.4em\relax Cambridge University Press, New York, NY, USA, 2011.

\bibitem{preskill2018quantum}
J.~Preskill, ``Quantum computing in the {NISQ} era and beyond,'' \emph{Quantum}, vol.~2, p.~79, 2018.

\bibitem{cerezo2021variational}
M.~Cerezo, A.~Arrasmith, R.~Babbush, S.~C. Benjamin, S.~Endo, K.~Fujii, J.~R. McClean, K.~Mitarai, X.~Yuan, L.~Cincio, and P.~J. Coles, ``Variational quantum algorithms,'' \emph{Nature Reviews Physics}, vol.~3, pp. 625--644, 2021.

\bibitem{farhi2014quantum}
E.~Farhi, J.~Goldstone, and S.~Gutmann, ``A quantum approximate optimization algorithm,'' \emph{{arXiv:1411.4028}}, 2014.

\bibitem{berberich2023quantum}
J.~Berberich and D.~Fink, ``Quantum computing through the lens of control: a tutorial introduction,'' \emph{arXiv:2310.12571}, 2023.

\bibitem{inoue2020model}
D.~Inoue and H.~Yoshida, ``Model predictive control for finite input systems using the {D}-wave quantum annealer,'' \emph{arXiv:2001.01400}, 2020.

\bibitem{deshpande2022quantum}
S.~A. Deshpande and A.~A. Kulkarni, ``The quantum advantage in decentralized control,'' \emph{arXiv:2207.12075}, 2022.

\bibitem{poljak1993checking}
S.~Poljak and J.~Rohn, ``Checking robust nonsingularity is {NP}-hard,'' \emph{Mathematics of Control, Signals, and Systems}, vol.~6, pp. 1--9, 1993.

\bibitem{nemirovskii1993several}
A.~Nemirovskii, ``Several {NP}-hard problems arising in robust stability analysis,'' \emph{Math. Control Signals Systems}, vol.~6, pp. 99--105, 1993.

\bibitem{blondel1997np}
V.~Blondel and J.~N. Tsitsiklis, ``{NP}-hardness of some linear control design problems,'' \emph{SIAM J. Control Optim.}, vol.~35, no.~6, pp. 2118--2127, 1997.

\bibitem{deif1986sensitivity}
A.~Deif, \emph{Sensitivity analysis in linear systems}.\hskip 1em plus 0.5em minus 0.4em\relax Springer-Verlag, Berlin, 1986.

\bibitem{doyle1982analysis}
J.~C. Doyle, ``Analysis of feedback systems with structured uncertainties,'' in \emph{Proc. IEEE}, 1982, pp. 242--250.

\bibitem{zhou1996robust}
K.~Zhou, J.~C. Doyle, and K.~Glover, \emph{Robust and optimal control}.\hskip 1em plus 0.5em minus 0.4em\relax Prentice-Hall, Inc., Englewood Cliffs, N.J., 1996.

\bibitem{rohn1994checking}
J.~Rohn, ``Checking positive definiteness or stability of symmetric interval matrices is {NP}-hard,'' \emph{Comment. Math. Univ. Carolin.}, vol.~35, no.~4, pp. 795--797, 1994.

\bibitem{hadfield2021representation}
S.~Hadfield, ``On the representation of boolean and real functions as hamiltonians for quantum computing,'' \emph{ACM Transactions on Quantum Computing}, vol.~2, no.~4, p.~18, 2021.

\bibitem{bergholm2018pennylane}
{V. Bergholm, J. Izaac, M. Schuld, C. Gogolin, S. Ahmed, \textit{et al.}}, ``Pennylane: Automatic differentiation of hybrid quantum-classical computations,'' \emph{arXiv:1811.04968}, 2018.

\bibitem{meisami2008dissipative}
M.~{Meisami-Azad}, J.~Mohammadpour, and K.~M. Grigoriadis, ``Dissipative analysis and control of state-space symmetric systems,'' in \emph{Proc. American Control Conf. (ACC)}, 2008, pp. 413--418.

\end{thebibliography}

\end{document}